\documentclass[runningheads]{llncs}
\usepackage[T1]{fontenc}
\usepackage{graphicx}
\usepackage{amssymb}
\usepackage{amsmath}
\usepackage[abs]{overpic}
\usepackage{tcolorbox}
\usepackage{graphicx,color,colordvi}
\usepackage{bbm}
\usepackage{cleveref}
\usepackage{stmaryrd}
\usepackage[utf8]{inputenc}
\usepackage{dsfont}
\usepackage{mathtools}

\usepackage{centernot}
\usepackage{pgfplots}

\def\openone{\leavevmode\hbox{\small1\kern-3.8pt\normalsize1}}

\def\11{\mathbb{I}}

\def\mc{\succeq_{\operatorname{m.c.}}}

\usepackage{pst-node}
\usepackage{tikz-cd}

\newcommand{\cA}{{\cal A}}

\newcommand{\cY}{{\cal Y}}

\def\d{\mathrm{d}}

\usepackage{graphicx}
\usepackage{setspace}
\usepackage{verbatim}
\usepackage{subfig}

\numberwithin{equation}{section}

\DeclareRobustCommand\openone{\leavevmode\hbox{\small1\normalsize\kern-.33em1}}

\newcommand{\be}{\begin{equation}}
	\newcommand{\ee}{\end{equation}}
\newcommand{\bea}{\begin{eqnarray}}
	\newcommand{\eea}{\end{eqnarray}}
\newcommand{\beas}{\begin{eqnarray*}}
	\newcommand{\eeas}{\end{eqnarray*}}

\DeclareFontFamily{U}{mathx}{\hyphenchar\font45}
\DeclareFontShape{U}{mathx}{m}{n}{<-> mathx10}{}
\DeclareSymbolFont{mathx}{U}{mathx}{m}{n}
\DeclareMathAccent{\widebar}{0}{mathx}{"73}

\newcommand{\Renyi}{R{\'e}nyi~}

\renewcommand{\d}{\textnormal{d}}
\newcommand{\cX}{{\cal X}}

\DeclareMathAccent{\widehat}{0}{mathx}{"70}
\DeclareMathAccent{\widecheck}{0}{mathx}{"71}

\begin{document}
\title{\Renyi partial orders for BISO channels}
%
\author{Christoph Hirche\orcidID{0000-0001-9265-827X}}
\authorrunning{C. Hirche}
%
\institute{Institute for Information Processing (tnt/L3S),\\ Leibniz Universit\"at Hannover, Germany \\
\email{hirche@tnt.uni-hannover.de}\\
\url{https://www.tnt.uni-hannover.de}}
\maketitle              
\begin{abstract}
A fundamental question in information theory is to quantify the loss of information under a noisy channel. Partial orders are typical tools to that end, however, they are often also challenging to evaluate. For the special class of binary input symmetric output (BISO) channels, Geng et al. showed that among channels with the same capacity, the binary symmetric channel (BSC) and binary erasure channel (BEC) are extremal with respect to the more capable order. Here we extend on this result by considering partial orders based on \Renyi mutual information. We establish the extremality of the BSC and BEC in this setting with respect to the generalized \Renyi capacity. In the process, we also generalize the needed tools and introduce $\alpha$-Lorenz curves. 

\end{abstract}
\section{Introduction}

Data processing inequalities state that information measures are monotone under the application of a noisy channel, i.e. information can only decrease along the transmission, qualifying them as distinguishability measures~\cite{polyanskiy2016}.
However, often we need to make more quantitative statements. To that end, information measure based partial orders allow to compare channels with respect to their data processing properties. A particular example is the more capable partial order~\cite{korner1975comparison}. With any such order it is an interesting question whether one can establish extremal channels among certain subclasses of channels.  
Geng et al.~\cite{geng2013broadcast} established that among all binary input symmetric output (BISO) channels with the same capacity the binary erasure channel (BEC) and the binary symmetric channel (BSC) form the two extremes with respect to the more capable partial order. This leads to their main result that the inner and outer bounds for the corresponding BISO broadcast channel differ if and only if the two BISO channels are more capable comparable. 
Recent work in~\cite{hirche2025partialOrders} extended on these results by considering the less noisy and degradable orders. In this case, capacity is not the natural common classifier, but several contraction coefficients play that role. 
Beyond those results, partial orders and contraction coefficients have found numerous further applications~\cite{asoodeh2020privacy,hirche2022bounding,xu2016information}.  
In this work, we take a different approach and consider generalizations of the more capable partial order using \Renyi mutual information.

\section{Preliminaries and technical tools}

\subsection{Notation}
In this work, we mainly consider binary input symmetric output (BISO) channels with input alphabet $\cX=\{0,1\}$ and output alphabet $\cY=\{0,\pm 1, \pm 2,\dots,\pm l\}$ for some integer $l\geq1$, those are the channels for which $P_{Y|X}(y|0)=P_{Y|X}(-y|1)\coloneqq p_y$. 
We can always assume that the output alphabet $\cY$ has even number of elements because else we can split $Y=0$ into two outputs, $Y = 0_+$  and $Y= 0_-$, with $P_{Y|X}(0_-|0) = P_{Y|X}(0_+|0) = \frac{p_0}{2}$. 
The binary convolution is denoted by $a\ast b=a(1-b)+(1-a)b$ and the binary \Renyi entropy is $h_\alpha(p)=\frac{1}{1-\alpha}\log\left(p^\alpha +(1-p)^\alpha\right)$, with the usual binary entropy as special case, $h(p)=\lim_{\alpha\to1}h_\alpha(p)= - p\log(p)-(1-p) \log(1-p)$. We denote the  Bernoulli distribution with probability $P(X=0) = p $ by $Ber(p)$.

\subsection{\Renyi mutual information}

The mutual information can be expressed as a Kulback-Leibler divergence in several ways. It is hence natural to seek a \Renyi generalization by instead using the \Renyi divergence~\cite{renyi1961measures}, 
\begin{align}
    D_\alpha(P\|Q) = \frac{1}{\alpha-1} \log \left( \sum_{a\in\cA} P^\alpha(a) Q^{1-\alpha}(a)  \right). 
\end{align}
Mostly for notational convenience, we also define the conditional \Renyi divergence as 
\begin{align}
    D_\alpha(P_{Y|X} \| Q_{Y|X} | P_X) = D_\alpha(P_{Y|X}P_X \| Q_{Y|X}P_X).
\end{align}
Sibson~\cite{sibson1969information} proposed the definition of an information radius which succinctly was generalized to that of a \Renyi mutual information by Verdu~\cite{verdu2015alpha}. That definition is
\begin{align}
    I^S_\alpha(X:Y) &= \min_{Q_Y} D_\alpha(P_{Y|X}\|Q_Y | P_X) \\
    &= \frac{\alpha}{\alpha-1} \log \sum_{y\in\cY} \left( \sum_{x\in\cX} P_X(x) P^\alpha_{Y|X=x}(y) \right)^{\frac1\alpha}. 
\end{align}
The closed form expression can be found in~\cite[Equation~(53)]{verdu2015alpha}.

For a given $P_{Y|X}$, the maximum mutual information over all input distributions $P_X$ is sometimes called the capacity of order $\alpha$, see~\cite{arimoto1977information}, and we denote it here by
\begin{align}
    C_\alpha(P_{Y|X}) = \sup_{P_X} I^S_\alpha(X:Y). 
\end{align}
For simplicity we simply call it the $\alpha$-capacity. 
In a slight abuse of notation, we denote by $BISO(C_\alpha)$ the set of all BISO channels with fixed capacity of order $\alpha$. 

Another definition of \Renyi mutual information was given by Arimoto~\cite{arimoto1977information} as
\begin{align}
    I^A_\alpha(X:Y) &= H_\alpha(X) - H^A_\alpha(X|Y) \\
    &= \frac{\alpha}{\alpha-1} \log \sum_{y\in\cY} \left( \sum_{x\in\cX} P_{X_\alpha}(x) P^\alpha_{Y|X=x}(y) \right)^{\frac1\alpha}, 
\end{align}
where $P_{X_\alpha}(x) = \frac{P^\alpha_{X}(x)}{\sum_x P^\alpha_{X}(x)}$. Although this is generally a very different quantity, the $\alpha$-capacity does not change using this definition~\cite{csiszar1995generalized}.

\subsection{\Renyi partial orders}

Körner and Marton~\cite{korner1975comparison} defined a channel $W_1:X\mapsto Y_1$ to be more capable than a channel $W_2:X\mapsto Y_2$, denoted $W_1 \mc W_2$, if 
\begin{align}
    I(X:Y_1) \geq I(X:Y_2)\quad\forall P_X. 
\end{align}
Here, we are interested in the generalization of this concept to \Renyi entropies and make the following definition. 
\begin{definition}
    A channel $W_1:X\mapsto Y_1$ is said to be $\alpha$-more capable than the channel $W_2:X\mapsto Y_2$, denoted $W_1 \mc^\alpha W_2$, if
    \begin{align}
        I^S_\alpha(X:Y_1) \geq I^S_\alpha(X:Y_2)\quad\forall P_X.
    \end{align}
\end{definition}
We will later see that for the purpose of this work we could have also used Arimotos \Renyi mutual information. 

\subsection{$\alpha$-Lorenz curves}
In~\cite{geng2013broadcast}, Geng et. al. gave a definition of the Lorenz curve for BISO channels. Here, we give an extension inspired by \Renyi entropies that we call the $\alpha$-Lorenz curve. We set $k_\alpha(p) = e^{(1-\alpha)h_\alpha(p)}$ and make the following definitions. 
\begin{definition}[$\alpha$-BISO partition and $\alpha$-BISO curve]
    For a BISO channel with transition probabilities $\{ p_y, p_{-y}\}_y$, rearrange $k_\frac{1}{\alpha}\left( \frac{p_y^\alpha}{p_y^\alpha + p_{-y}^\alpha}\right)$ in ascending order and denote the permutation by $\pi$. The $\alpha$-BISO partition is defined as the partition of $[0,d_C]$ with points \begin{align}
        t_k=\sum_{i=1}^k (p_{\pi_i}^\alpha + p_{-\pi_i}^\alpha)^\frac{1}{\alpha}
    \end{align} and $d_C=\sum_{y>0} (p_y^\alpha + p_{-y}^\alpha)^\frac{1}{\alpha}$. We set $t_0=0$. The $\alpha$-BISO curve is defined as the stepwise function $f_\alpha(t)$ such that 
    \begin{align}
        f_\alpha(t)
        =k_\frac{1}{\alpha}\left( \frac{p_{\pi_y}^\alpha}{p_{\pi_y}^\alpha + p_{-{\pi_y}}^\alpha}\right)
    \end{align} on $(t_{k-1},t_k]$, and $f_\alpha(0)=0$. 
\end{definition}
Note that the interval $[0,d_C]$ on which we define the partition depends on the channel. However, we will see later that the $d_C$ is solely determined by the capacity of the channel and therefore the same for all BISO channels with identical capacity. 
Equipped with this, we can make the following definition. 
\begin{definition}[$\alpha$-Lorenz curve of a BISO channel]
    For a BISO channel with $\alpha$-BISO curve $f_\alpha(t)$, the $\alpha$-Lorenz curve $F_\alpha(t)$ is defined as \begin{align}
        F_\alpha(t) 
        = \int_0^t f_\alpha(\tau)
        \d\tau.
    \end{align} 
\end{definition}
The following result holds for $\alpha$-BISO curves. 
\begin{lemma}\label{Lem:Lorenz-curve}
    Given BISO channels $X\to Y$ and $X\to Z$ with equal capacity of order $\alpha$ and $\alpha$-BISO curves $f_\alpha(t)$ and $g_\alpha(t)$, respectively. Let the common refinement of these two BISO partitions be $\{t_k : k=0,\dots,N\}$, and $\xi_k=t_k-t_{k-1}$. Then
    \begin{align}
        F_\alpha(t_i) = \sum_{k=1}^i \xi_k f_\alpha(t_k) \leq \sum_{k=1}^i \xi_k g_\alpha(t_k) = G_\alpha(t_i), \quad i=1,\dots,N
    \end{align}
    if and only if the $\alpha$-Lorenz curve satisfies  $F_\alpha(t)\leq G_\alpha(t)$ for all $t\in[0,d_C]$.  
\end{lemma}
\begin{proof}
    The proof is essentially the same as that of~\cite[Lemma 1]{geng2013broadcast} and essentially uses the observation that $F_\alpha(t)$ is a piecewise linear function. 
\end{proof}
We are now in the position to combine all concepts of this section and make statements about \Renyi partial orders for BISO channels. 

\section{Main results}

In this section we will only consider BISO channels. After some derivation, one can find that for a BISO channel $P_{Y|X}$, Sibson's \Renyi mutual information becomes
\begin{align}
    I_\alpha^S(X:Y) &= \frac{\alpha}{\alpha-1} \log \sum_{y>0} (p_y^\alpha + p_{-y}^\alpha)^\frac{1}{\alpha} k_\frac{1}{\alpha}\left(x\ast \frac{p_y^\alpha}{p_y^\alpha + p_{-y}^\alpha}\right) \\
    &= \frac{\alpha}{\alpha-1} \log \int_0^{d_C} k_\frac{1}{\alpha}\left(x\ast k^{-1}_\frac{1}{\alpha}(f(\tau)) \right) \d\tau,  
\end{align}
where $x:=P(X=0)$ and we recall $k_\alpha(p) = e^{(1-\alpha)h_\alpha(p)}$. From here one can find that the $\alpha$-capacity is given by
\begin{align}
    C_\alpha(P_{Y|X}) = \log2 - \frac{\alpha}{\alpha-1} \log \sum_{y>0} (p_y^\alpha + p_{-y}^\alpha)^\frac{1}{\alpha},  
\end{align}
that is, it is optimized by $x=\frac12$. This follows because $k_\frac{1}{\alpha}(x)$ is maximal at $x=\frac12$ and $\frac12\ast x=\frac12$. 
This implies $C_\alpha(P_{Y|X}) = \log2 - \frac{\alpha}{\alpha-1} \log d_C$, showing our earlier claim that $d_C$ depends solely on the $\alpha$-capacity. 

We will need the following lemma, previously proven in the context of information combining. 
\begin{lemma}[Lemma IV.10 in~\cite{hirche2020renyi}]\label{Lem:conv-conc}
    The function
    \begin{align}
        k_{\alpha}\left(k^{-1}_{\alpha}(x)\ast k^{-1}_{\alpha}(y) \right)
    \end{align}
    is convex in both $x$ and $y$ for $0<\alpha<1$ and $2<\alpha\leq 3$ and concave for $1<\alpha\leq2$ and $\alpha\geq 3$. 
\end{lemma}
In particular, note that for $\alpha\in\{2,3\}$ the function is indeed linear in both $x$ and $y$. Again similar to~\cite{geng2013broadcast}, we will also need the following lemma. 
\begin{lemma}[Lemma 1 in~\cite{hajek1979evaluation}]\label{Lem:convex-order}
    Let $x_1,\dots,x_l$ and $y_1,\dots,y_l$ be nondecreasing sequences of real numbers. Let $\xi_1,\dots,\xi_l$ be a sequence of real numbers such that 
    \begin{align}
        \sum_{j=k}^l \xi_j x_j \geq \sum_{j=k}^l \xi_j y_j,\quad 1\leq k\leq l, 
    \end{align}
    with equality for $k=1$. Then for any convex function $\Lambda$, 
    \begin{align}
        \sum_{j=1}^l \xi_j \Lambda(x_j) \geq \sum_{j=1}^l \xi_j \Lambda(y_j). \label{Eq:convex-ineq}
    \end{align}
\end{lemma}
It is easy to check that in the above lemma if $\Lambda$ is concave instead, then Equation~\eqref{Eq:convex-ineq} holds with the direction of the inequality swapped. 

With these result, we can establish the following theorem. 
\begin{theorem}[A sufficient condition]\label{Thm:LorenzToMC}
    Given BISO channels $W_1:X\mapsto Y_1$ and $W_2:X\mapsto Y_2$ with same $\alpha$-capacity $C_\alpha$ and $\alpha$-Lorenz curves $F_{\alpha,1}(t)$ and $F_{\alpha,2}(t)$, respectively. If $F_{\alpha,1}(t)\leq F_{\alpha,2}(t)$ then $W_1$ is $\alpha$-more capable than $W_2$, i.e. $W_1 \mc^\alpha W_2$, for all $\alpha>1$, $\frac12\leq\alpha<1$ and $0<\alpha\leq\frac13$. For $\frac{1}{3}\leq\alpha\leq\frac12$ we have instead $W_2 \mc^\alpha W_1$. 
\end{theorem}
\begin{proof}
    Since both channels have the same capacity of order alpha, their Lorenz curves are defined on the same interval $[0,d_C]$. By Lemma~\ref{Lem:Lorenz-curve}, $F_{\alpha,1}(t)\leq F_{\alpha,2}(t)$ implies 
    \begin{align}
        F_{\alpha,1}(t_i) = \sum_{k=1}^i \xi_k f_{\alpha,1}(t_k) \leq \sum_{k=1}^i \xi_k f_{\alpha,2}(t_k) = F_{\alpha,2}(t_i), \quad i=1,\dots,N.
    \end{align}
    Next, we have, 
    \begin{align}
         F_{\alpha,1}(d_C) = \int_0^{d_C} f_{\alpha,1}(\tau) \d\tau &= \sum_{i=1}^N (p_{\pi_i}^\alpha + p_{-\pi_i}^\alpha)^\frac{1}{\alpha}k_\frac{1}{\alpha}\left( \frac{p_{\pi_y}^\alpha}{p_{\pi_y}^\alpha + p_{-{\pi_y}}^\alpha}\right) \\
         &= \sum_{i=1}^N (p_{\pi_i} + p_{-\pi_i}) = 1. 
    \end{align}
    This holds indeed for all BISO channels, and therefore we have $F_{\alpha,1}(d_C)=F_{\alpha,2}(d_C)=1$. Furthermore, $f_{\alpha,1}(t_k)$ and $f_{\alpha,2}(t_k)$ are both nondecreasing and we can apply Lemma~\ref{Lem:convex-order} with the function $\Lambda(y)=k_{\alpha}\left(x\ast k^{-1}_{\alpha}(y) \right)$. Due to the convexity properties in Lemma~\ref{Lem:conv-conc}, this gives for $\alpha>1$ and $\frac{1}{3}\leq\alpha\leq\frac12$, 
    \begin{align}
        \sum_{j=1}^N \xi_j k_\frac{1}{\alpha}\left(x\ast k^{-1}_\frac{1}{\alpha}(f_{\alpha,1}(t_j)) \right) \geq \sum_{j=1}^N \xi_j k_\frac{1}{\alpha}\left(x\ast k^{-1}_\frac{1}{\alpha}(f_{\alpha,2}(t_j)) \right), 
    \end{align}
    and for $\frac12\leq\alpha<1$ and $0<\alpha\leq\frac13$, 
    \begin{align}
        \sum_{j=1}^N \xi_j k_\frac{1}{\alpha}\left(x\ast k^{-1}_\frac{1}{\alpha}(f_{\alpha,1}(t_j)) \right) \leq \sum_{j=1}^N \xi_j k_\frac{1}{\alpha}\left(x\ast k^{-1}_\frac{1}{\alpha}(f_{\alpha,2}(t_j)) \right). 
    \end{align}
    Taking $\frac{\alpha}{\alpha-1}\log(\cdot)$ on both sides, noting the negative prefactor for $\alpha<1$, we have for $\alpha>1$, $\frac12\leq\alpha<1$ and $0<\alpha\leq\frac13$, that
    \begin{align}
        I^S_\alpha(X:Y_1) \geq I^S_\alpha(X:Y_2), \quad \forall p(x), 
    \end{align}
    and for $\frac{1}{3}\leq\alpha\leq\frac12$ with the inequality exchanged. This implies the claimed result.     
\end{proof}
For large ranges of $\alpha$ this shows the same behavior as the $\alpha=1$ case discussed in~\cite{geng2013broadcast}. The notable exception is $\frac{1}{3}\leq\alpha\leq\frac12$, where the implication is turned around. However, special attention should also be given to the settings where $\alpha\in\{\frac13,\frac12\}$. As here both directions hold, ordered $\alpha$-Lorenz curves imply that the channels are equivalent in the $\alpha$-more capable ordering. Of course, then one might expext this to hold for any two channels, independent of their Lorenz curves. To that end, a calculation can show that, 
\begin{align}
    I_\frac{1}{3}^S(X:Y) &= -\frac{1}{2} \log\left[ 1-x(1-x)\left[4-\sum_{y>0} (p_y^\frac{1}{3} + p_{-y}^\frac{1}{3})^3 \right]\right] \\
    &= -\frac{1}{2} \log\left[ 1-x(1-x)\left[4-d_{C_\frac{1}{3}} \right]\right] \\
    I_\frac{1}{2}^S(X:Y) &= - \log\left[ 1-2x(1-x)\left[2-\sum_{y>0} (p_y^\frac{1}{2} + p_{-y}^\frac{1}{2})^2 \right]\right] \\
    &= - \log\left[ 1-2x(1-x)\left[2-d_{C_\frac{1}{2}} \right]\right]. 
\end{align}
Hence, in terms of channel properties, the mutual information only depends on the $\alpha$-capacity of the channel. This gives immediately the following corollary. 
\begin{corollary}
    For any two BISO channels $W_1:X\mapsto Y_1$ and $W_2:X\mapsto Y_2$ with same $\alpha$-capacity $C_\alpha$, we have
    \begin{align}
        W_1 \mc^\frac{1}{3} W_2, \quad W_2 \mc^\frac{1}{3} W_1, \\ 
        W_1 \mc^\frac{1}{2} W_2, \quad W_2 \mc^\frac{1}{2} W_1. 
    \end{align}
\end{corollary}
It follows also more generally that for these values of $\alpha\in\{\frac12,\frac13\}$ the $\alpha$-more capable order is equivalent to comparing the $\alpha$-capacity between channels. 
Beyond these values, we find that, similar to the $\alpha=1$ case in~\cite{geng2013broadcast}, the binary erasure channel and the binary symmetric channel are the extrema with respect to the $\alpha$-more capable order. 
\begin{corollary}
    Let, $W$, $BSC$ and $BEC$ be channels with the same $\alpha$-capacity. Then, we have for $\alpha>1$, $\frac12\leq\alpha<1$ and $0<\alpha\leq\frac13$
    \begin{align}
        BEC \mc^\alpha W \mc^\alpha BSC, 
    \end{align}
    and for $\frac{1}{3}\leq\alpha\leq\frac12$, 
    \begin{align}
       BSC  \mc^\alpha W \mc^\alpha BEC. 
    \end{align}
\end{corollary}
\begin{proof}
    By Theorem~\ref{Thm:LorenzToMC} it is sufficient to prove that the $\alpha$-Lorenz curves are ordered. To that end, recall that the $\alpha$-Lorenz curve is always a convex function. For the BSC, the only BISO channel with output dimension $2$, the curve is simply the straight line from $0$ to $1$ on the interval $[0,d_C]$ and therefore an upper bound on any convex function with the same endpoints. The BEC is slightly more complicated. Consider $\alpha\geq1$, then the curve starts with a derivative of $1$ on the interval $(0,1-\epsilon]$, which is indeed the smallest possible derivative in this case. Then, on the interval $(1-\epsilon, 1-\epsilon + \epsilon 2^{\frac{1-\alpha}{\alpha}}]$, it has the maximal derivative $2^\frac{\alpha-1}{\alpha}$. Hence, the BEC curve always gives a lower bound. A similar argument holds for $\alpha<1$ but now the maximum derivative is $1$ and the minimum derivative is $2^\frac{\alpha-1}{\alpha}$. 
\end{proof}

\subsection{Arimoto mutual information}
So far, all derived results have been proven for the Sibson mutual information definition. Here we briefly point out that all of the above results also hold for Arimoto's mutual information. Note that for BISO channels, we get
\begin{align}
    I_\alpha^A(X:Y) &= \frac{\alpha}{\alpha-1} \log \sum_{y>0} (p_y^\alpha + p_{-y}^\alpha)^\frac{1}{\alpha} k_\frac{1}{\alpha}\left(x_\alpha \ast \frac{p_y^\alpha}{p_y^\alpha + p_{-y}^\alpha}\right) \\
    &= \frac{\alpha}{\alpha-1} \log \int_0^{d_C} k_\frac{1}{\alpha}\left(x_\alpha\ast k^{-1}_\frac{1}{\alpha}(f(\tau)) \right) \d\tau,  
\end{align}
with $x_\alpha=\frac{x^\alpha}{x^\alpha + (1-x)^\alpha}$. Hence, we can use the same definition for $\alpha$-Lorenz curves and now use the convexity or concavity of 
    \begin{align}
        k_{\alpha}\left(x_\alpha\ast k^{-1}_{\alpha}(y) \right),
    \end{align}
 which still follows from Lemma~\ref{Lem:conv-conc} because we keep $\alpha$ fixed. 

The Arimoto mutual information can be of particular interest, because it obeys a chain rule, allowing to translate the results also to conditional entropies,
\begin{align}
    &I^A_\alpha(X:Y_1) \geq I^A_\alpha(X:Y_2)  \Leftrightarrow H^A_\alpha(X|Y_1) \leq H^A_\alpha(X|Y_2).     
\end{align}
A similar result does not directly holds for the \Renyi mutual information of Sibson.

\section{Conclusions}

In this work, we establish extremality result for a partial order based on \Renyi mutual information. The results are closely connected to \Renyi bounds on information combining~\cite{hirche2020renyi,hirche2023chain}. A natural direction for further investigation would be to go beyond BISO channels, include auxiliary random variables akin to the less noisy ordering or to consider quantum settings, where it might help to solve open problems in information combining~\cite{Hirche_2018}. One might expect that such results will have implications for contraction coefficients of \Renyi divergences~\cite{grosse2025strong}.

\begin{credits}
\subsubsection{\ackname} Funded by the Deutsche Forschungsgemeinschaft (DFG, German
Research Foundation) – 550206990. 

\subsubsection{\discintname}
The authors have no competing interests to declare that are
relevant to the content of this article. 
\end{credits}
%
%
%
\bibliographystyle{splncs04} \bibliography{lib}

\end{document}